\documentclass[11pt]{article}
\usepackage{geometry}
\usepackage{graphicx}
\usepackage{amsmath}
\usepackage{amssymb}
\usepackage[all]{xy}

\usepackage[dvipsnames]{xcolor}

\newtheorem{theorem}{Theorem}[section]
\newtheorem{corollary}[theorem]{Corollary}
\newtheorem{lemma}[theorem]{Lemma}
\newtheorem{proposition}[theorem]{Proposition}
\newtheorem{claim}[theorem]{Claim}
\newtheorem{definition}[theorem]{Definition}

\newtheorem{example}[theorem]{Example}
\newtheorem{observation}[theorem]{Observation}

\def\squarebox#1{\hbox to #1{\hfill\vbox to #1{\vfill}}}
\newcommand{\qed}{\hspace*{\fill}\vbox{\hrule\hbox{\vrule\squarebox{.667em}\vrule}\hrule}\smallskip}
\newenvironment{proof}{\noindent{\bf Proof:~~}}{\(\qed\)}

\newcommand{\xhdr}[1]{\paragraph{\bf #1.}}
\newcommand{\comment}[1]{}

\newcommand{\eps}{\varepsilon} 

\newcommand{\expect}{{\mathbb{E}}}

\begin{document}
	\title{Optimal Stopping with Behaviorally Biased Agents: \\The Role of Loss Aversion and Changing Reference Points}
	
	\author{Jon Kleinberg
		\thanks{
			Cornell University.
			Email: kleinber@cs.cornell.edu.
			Research supported in part by a Vannevar Bush Faculty Fellowship, MURI grant W911NF-19-0217, 
			AFOSR grant FA9550-19-1-0183, and BSF grant 2018206.
		}
		\and Robert Kleinberg
		\thanks{
			Cornell University.
			Email: rdk@cs.cornell.edu.
			Research supported by NSF Grant CCF-1512964.
		}
		\and
		Sigal Oren
		\thanks{
			Ben-Gurion University of the Negev.
			Email: sigal3@gmail.com. 
			Work supported by BSF grant 2018206 and ISF grant 2167/19.
		}
	}
	\date{May, 2020}
	
	\begin{titlepage}
		\maketitle
		
		\begin{abstract}
			People are often reluctant to sell a house, or shares of stock, 
below the price at which they originally bought it.
While this is generally not consistent with rational utility maximization,
it does reflect two strong empirical regularities that are
central to the behavioral science of human decision-making:
a tendency to evaluate outcomes relative to a {\em reference point}
determined by context (in this case the original purchase price),
and the phenomenon of {\em loss aversion} in which people are
particularly prone to avoid outcomes below the reference point.
Here we explore the implications of reference points and loss aversion
in optimal stopping problems, where people evaluate a sequence of
options in one pass, either accepting the option and stopping the search
or giving up on the option forever.
The best option seen so far sets a reference point that shifts
as the search progresses, and a biased decision-maker's utility
incurs an additional penalty when they accept a later option 
that is below this reference point.

We formulate and study a behaviorally well-motivated version of the 
optimal stopping problem that incorporates these notions of
reference dependence and loss aversion.
We obtain tight bounds on the performance of
a biased agent in this model
relative to the best option obtainable in retrospect
(a type of {\em prophet inequality} for biased agents), as well as
tight bounds on the ratio between the performance of a biased agent
and the performance of a rational one.
We further establish basic monotonicity results, and show an
exponential gap between the performance of a biased agent in a
stopping problem with respect to a worst-case versus a random order.
As part of this, we establish fundamental differences 
between optimal stopping problems for rational versus biased agents,
and these differences inform our analysis.

		\end{abstract}
		
		\thispagestyle{empty}
	\end{titlepage}

\section{Introduction} One of the central human biases studied in behavioral economics is
{\em reference dependence} --- people's tendency to evaluate an
outcome not in absolute terms but instead relative
to a {\em reference point} that reflects some notion of the status quo
\cite{kahneman1979prospect}.
Reference dependence interacts closely with a related
behavioral bias, {\em loss aversion}, in which people weigh
losses more strongly than gains of comparable absolute values.
Taken together, these two effects produce a fundamental behavioral
regularity in human choices: once a reference point has been
established, people tend to avoid outcomes in which
they experience a loss relative to the reference point.

This effect can be seen in simple examples.
One well-known instance of the effect is the empirical evidence that
individual investors will tend to avoid selling a stock unless it
has exceeded the price at which they purchased it; the purchase
price thus acquires a special status in the investor's decision-making
\cite{odean1998investors}.
In this example, the purchase price $x$ serves as the reference point,
and the investor's loss aversion relative to this reference point
leads them to hold on to a stock while its value is below $x$.
There are two useful points to note about this example, and
reference dependence more generally.
First, it is a deviation from rational behavior: once the investor
has purchased the stock, the value of the purchase price is
irrelevant to any future decisions, which should be based purely
on predictions of the stock's future performance.
Second, the reference point
will often be set by some notion of status quo or default that is
determined by the semantics of the environment.
Once we realize, for example,
that the purchase price serves naturally as a reference
point for stocks, we might conjecture that a similar effect should
hold for other items, like houses --- that
home-owners should exhibit different behaviors in situations where
they are contemplating selling their house at a value
either above or below the price at which they originally purchased it.
This effect too is borne out in empirical data
\cite{genesove2001loss}.

In more complex examples, and most relevant to our work here,
the reference may shift while an agent is making a decision.
Consider for example an agent who is trying to make a large purchase
or hire a job candidate, and does this by evaluating candidate options
in one pass in a take-it-or-leave-it fashion --- with each candidate
they must either accept it and end the search, or give up on
it as an option forever.
Experimental studies by Schunk and Winter \cite{schunk2009relationship}
show that people in this type of task behave
consistently with the notion that they are maintaining a time-varying
reference point equal to the best option they have seen so far.
This means that if they settle for a candidate $A$ that is worse than
a better candidate $B$ that they have seen in the past,
their utility from selecting $A$ will be reduced by some notion of
loss relative to the high reference point set by $B$.
In these studies, people's decisions are best explained by a model
in which they take into account this anticipated sense of
loss prospectively in making their choices; they operate so as to
reduce the chance that they will have to choose a future option
that is dominated by one that they have earlier passed up.
We can recognize this empirical regularity at an intuitive level as well:
if we purchase a house that is worse than one that we passed up on
earlier in our search, or take a job that is worse than one we
were offered earlier and turned down,
we tend to feel that we have experienced
a loss relative to what we could have achieved.
Our utility for the final outcome thus has a history-dependence, in
that it is affected by the options that we rejected prior to our
ultimate choice.

\xhdr{Implications for Optimal Stopping Problems}
These findings have interesting implications for how we might
model the ways in which human decision-makers with standard biases
approach optimal stopping problems.
Consider one of the canonical formulations of an optimal stopping problem,
which closely follows the structure of the Schunk-Winter
experiments described above:
we are presented with a sequence of $n$ candidates in order;
we know that candidate $t$ will have a value $v_t$ drawn from a
distribution $\mathcal F_t$; but we only see the materialized value of
$v_t$ when we get to candidate $t$ in order.
At that point, we must either accept $t$ and end the search, or give up
on $t$ forever.
The celebrated {\em prophet inequality} of Krengel and Sucheston \cite{krengel1978semiamarts}
asserts that there is an on-line policy for choosing one of the
candidates in such a way that the expected value of the selected
candidate is within a factor of two of the highest value of any
of the candidates when viewed in retrospect (i.e. the performance achieved
by a {\em prophet} that can see all the values in advance).
An active line of algorithmic research has developed to understand
the bounds that can be achieved in models of this form;
see the survey articles \cite{correa2019recent,lucier2017economic} and
the references therein.

Given the behavioral literature on reference dependence, it becomes
interesting to ask how the model should change to incorporate this bias.
In particular, consider this type of stopping problem being solved by an
agent with human biases based on reference dependence.
Let $\lambda \geq 0$
be a parameter representing the agent's level of loss aversion;
$\lambda$ is a multiplier on any losses the agent experiences relative
to a reference point.\footnote{We choose the natural, albeit simple, way of modeling loss aversion as linear in the loss to keep the focus on the role that the changing reference point plays. As part of future work it will be interesting to study the effects of different modelings of the loss function on the behavior of the agent. }
Now, when the agent contemplates a candidate of value $u$ in their
one-pass search and the value of the best candidate they have seen so far is $v$,  their utility from selecting $u\geq v$ will be equal to $u$, however for $u<v$ their utility will be
$$u - \lambda(v - u).$$
We will refer to an agent with this behavior as a
{\em reference-dependent agent} with loss-aversion $\lambda$.
Thus, $\lambda = 0$ corresponds to the traditional optimal stopping
problem with a rational agent, while larger values of $\lambda$
correspond to greater levels of loss aversion.
The experimental results suggest that we study the setting in which
agents are {\em sophisticated} about their reference dependence, in that
they make choices with the awareness that their utilities will be shifted
by the $\lambda(v - u)$ term. Thus, a reference-dependent agent 
	will choose a stopping rule that maximizes its expected ``shifted'' utility.  

\xhdr{The present work: Optimal stopping with reference dependence}
We thus have a basic new model of sequential decision-making
that is behaviorally well-motivated, incorporating one of the central
biases in human decision-making.
Our goal in the present work is to explore this model; we
provide both quantitative bounds for the performance of
reference-dependent agents relative to optimality, including a tight
prophet inequality for biased agents, as well as
more qualitative analysis of the model's behavior, including
questions about its monotonicity in the loss-aversion parameter $\lambda$
and the length of the sequence of candidates.

A brief overview of our results is as follows.
First, we show that there is a strategy by which a
reference-dependent agent with loss-aversion $\lambda$ can guarantee
an expected value within a factor of $\lambda + 2$ of the
best candidate in retrospect.
This can be viewed as a type of prophet inequality for biased agents,
distinct from other results of this form, which assume
the traditional rational model of decision-making.
We show that the bound can be obtained by a threshold rule,
analogous to Samuel-Cahn's approach to the classical prophet inequality
\cite{samuel1984comparison}
but with the threshold shifted by a parameter based on the bias.
The bound is tight for all $\lambda$, and it recovers the factor of 2 from
the classical prophet inequality as $\lambda$ goes to $0$.
A different comparison is between the performance of
a reference-dependent agent with loss-aversion $\lambda$
and a rational agent that also must operate on-line;
here we show that there is a strategy by which the reference-dependent
agent can guarantee an expected value within a factor of
$\lambda + 1$ of the rational agent, and this too is tight.
Essentially, the reason for these large gaps is that
reference-dependent agents with loss-aversion tend to 
select candidates too early in the sequence since they are more
worried about ending up with a candidate with lower value.

A growing line of work has studied the effect of random arrival order
on optimal stopping problems --- when the candidates are seen according
to an order selected from the uniform distribution on permutations.
For rational agents (without a notion of reference dependence) this
has been shown to improve the bound of 2 to a smaller constant
strictly between 1 and 2; after a sequence of improvements the best known bound is slightly less than
	$\frac{3}{2}$ and is due to Correa et al. \cite{correa2020prophet}.
For reference-dependent agents, we show that random ordering has
a large effect on the achievable bounds as a function of $\lambda$;
under a random ordering, a reference-dependent agent can achieve
an expected value that the optimum exceeds by a factor of only
$\Theta(\log \lambda)$, an exponential improvement over the
$\lambda + 2$ bound for worst-case orderings.
We observe an additional improvement by order of
magnitude when the distributions have only two values in their
support and the ordering of the candidates is the one that maximizes
the expected value of the candidate that the agent selects. For this case, we improve the bound to $2$.
This result is analogous to the result of \cite{agrawal2020optimal} showing that
for rational agents and $2$-point distributions the bound is $5/4$.


As noted above, we prove a number of monotonicity results showing
that the performance of biased agents degrades in certain
consistent directions.
First, suppose two reference-dependent agents have loss-aversion
parameters $\lambda$ and $\lambda'$, with $\lambda < \lambda'$.
Then we can show that the expected value of the candidate selected by the agent
with the lower parameter $\lambda$ is higher than the expected
value obtained by the agent with the higher parameter $\lambda'$.
To prove this we show that as the value of $\lambda$ increases the
agent tends to select candidates that are earlier
in the sequence of candidates. Second, we establish monotonicity in the sequence of candidates:
adding a candidate to the end of a sequence cannot
hurt a biased agent's performance, but
adding a candidate to the beginning of the sequence can reduce
the expected value of the candidate selected by the agent (by a tight factor of
$\lambda + 1$).
A result like this has no natural analogue for rational agents,
where adding options cannot hurt performance.

The analysis underlying these results highlights a number of
ways in which optimal stopping
questions for biased agents become qualitatively
different from their traditional analogues for unbiased agents.
Many of these flow from the initial observation that a
biased agent's optimal decisions in these problems are
history-dependent, since they depend on the maximum among
the realized values earlier in the sequence;
in contrast, a key feature of optimal stopping problems with
rational agents is their memoryless nature.
The three-way comparison between the optimal value (achieved by a prophet),
the value achieved by a rational agent,
and the value achieved by a reference-dependent agent
also highlights important aspects of the model,
including the fact that the largest of these worst-case gaps ---
the factor of $\lambda + 2$ between the optimal value and the
biased agent --- is strictly less than the product of the other
two gaps, suggesting that the worst cases are obtained at
fundamentally different points.
The behavior of biased agents under random orderings also
highlights new aspects of the model --- both in the exponential
improvement it yields in $\lambda$, and also in the observation
that the gap between the biased agent's value and the optimal value
can always be upper-bounded by the sequence length.
In contrast, no upper bound depending only on the sequence length
is possible for biased agents under a worst-case order.

In the next section, we present the model and accompanying
definitions and notations in full detail, and then provide proofs
for the results stated here.

\section{Model and Preliminaries} \label{sec-model} Consider the following selection problem. There are $n$ time steps. At each time step $t$ a new candidate arrives and its value $v_t$ is sampled from some known distribution $\mathcal F_t$ with non-negative support. We use $V_t$ to denote a random variable equals to the value of candidate $t$ and $v_t$ to denote its realized value. When candidate $t$ arrives its value $v_t$ is realized and the agent should decide whether to select the candidate or not. Selection decisions are final. That is, if an agent decided not to select a candidate it cannot change its mind later on. We assume that the agent has to select some candidate, thus, if the agent reaches the last candidate it would always select it. This is a natural assumption for a hiring scenario in which a manager has to recruit an employee to fill a vacancy. It is also applicable in many other settings, for example, a buyer that has to buy a car.

We consider an agent that has loss aversion with respect to the value of the best candidate it has considered so far. In particular, if the best candidate that was considered has value $v$ and the candidate that was selected has a value $u<v$ then the total utility of the loss averse agent would be $u-\lambda (v-u)$. Here $\lambda\geq 0$ is a parameter capturing the extent of loss aversion. $\lambda=0$ implies that the agent is completely rational.\footnote{We only consider non-negative values of $\lambda$ as a negative value of $\lambda$ implies an agent that prefers to select a candidate worse than the best it has considered, which makes less sense.} We refer to an agent that is loss averse with parameter $\lambda$ as a $\lambda$-biased agent and to the value of the best candidate so far as the \emph{reference value}.

Observe that the loss averse agent's planning problem can be modeled as a finite-horizon Markov decision process (MDP) in which the state encodes the reference value and the number of candidates considered. As such, the MDP value function satisfies the Bellman Equation and can be computed by dynamic programming. From the dynamic program we can easily infer the optimal stopping rule for a $\lambda$-biased agent. Throughout the paper we refer to this optimal stopping rule as the \emph{optimal $\lambda$-biased stopping rule}.

To help the reader get acquainted with the model, we briefly work out the details of the dynamic program: let $m$ be the number of possible values for candidates (i.e., all values that are in the support of any of the distributions $\mathcal F_t$). The table size is $m\times n$ and $U_{\lambda}[v,t]$ is the expected utility of playing optimally at times $t, \ldots,n$ when the reference value is $v$ (i.e., $v=\max\{v_1,\ldots,v_{t-1}\}$). With this notation the agent selects candidate $t$ if and only if:
\begin{align*}
v_t - \lambda(v-v_t)^+ \geq U_\lambda[v \lor v_t,t+1].
\end{align*}
where $x \lor y=\max\{x,y\}$ and $(x)^+$ is $x$ for $x\geq 0$ and $0$ otherwise. This gives us the following formula for the dynamic program:
\begin{align*}
U_\lambda[v,t]:=\expect_{v_t  \sim \mathcal F_t} [\underbrace{U_\lambda[v \lor v_t,t+1]}_{v_t \text{~is not selected}} \lor \underbrace{(v_t - \lambda(v-v_t)^+)}_{v_t \text{~is selected}} ].
\end{align*}

Recall that if the agent reaches the last candidate it has to select it. Thus, we have that $U_\lambda[v,n]:=\expect_{v_n  \sim \mathcal F_n}[v_n - \lambda(v-v_n)^+]$ for any value $v$ that is in the support of some distribution. We can solve the dynamic program by continuing to fill the cells backwards. Observe that $U_{\lambda}[0,1]$ holds the expected utility of the optimal $\lambda$-biased stopping rule as in the beginning the agent's reference value is $0$.

For agents that do not exhibit a bias it is well known  that the optimal stopping rule can be described by a sequence of monotonically decreasing thresholds $\tau_1,\ldots, \tau_n$ such that a candidate $t$ is selected if and only if $v_t \geq \tau_t$. (See, for example, Theorem 3.2 of the textbook \cite{crs-book}.) In Section \ref{sec-mon} we show that the optimal $\lambda$-biased stopping rule can also be described as a sequence of thresholds, albeit the thresholds for biased agents are history dependent.  To prove this claim we observe a monotonicity property: for reference values $v'>v$ and any $t\geq 1$ we have that $U_\lambda[v',t] \leq U_\lambda[v,t]$. We prove this observation and observe other types of monotonicity in Section \ref{sec-mon}.

\section{The Cost of Loss-Aversion} \label{sec-bound} There are two interesting comparisons one can make here. First, we compare between the expected value of a candidate that a biased agent selects and the value of the best candidate in hindsight (what a prophet would have chosen). The second comparison is between the expected value of the candidate that was selected by a $\lambda$-biased agent in comparison with the value of the candidate selected by an unbiased agent.  We focus on comparing the value of the selected candidate and not the utility of the agent, as the value of the selected candidate is an objective measure and hence is more appropriate for quantifying the actual loss (as oppose to the perceived loss) of a $\lambda$-biased agent due to its bias.

We define random variables $V^*$ and $V_\lambda$ (for $\lambda\geq0$) to equal the value of the best candidate in hindsight and the candidate selected by a $\lambda$-biased agent using the optimal stopping rule respectively.
%
 We show that the ratio between $\expect[V^*]$ and $\expect[V_{\lambda}]$ can be as high as $\lambda +2$ and the ratio between $\expect[V_{0}]$ (i.e., the expected value of a candidate selected by an unbiased agent) and $\expect[V_{\lambda}]$ can be as high as $\lambda+1$. This is done by adapting a classic example for the prophet inequality.


\begin{example} \label{exm-ratio}
	For $\eps>0$, let $V_1 = \frac{1}{1+(1-\eps)\lambda}$ and
	$V_2= \begin{cases}
	\frac{1}{\eps}&\text{w.p $\eps$}\\
	0&\text{w.p $1-\eps$}
	\end{cases}$.
\end{example}
\begin{claim}
For every $\lambda>0$, there exists an instance such that the ratio $\frac{\expect[V^*]}{\expect[V_\lambda]}$ is arbitrarily close to $\lambda+2$.
\end{claim}
\begin{proof}
	Consider the family of instances defined in Example \ref{exm-ratio}. A $\lambda$-biased agent always selects the first candidate. To see why this is the case observe that:
	\begin{align*}
	\frac{1}{1+(1-\eps)\lambda} \geq 1 -\lambda(1-\eps) \cdot \left( \frac{1}{1+(1-\eps)\lambda} \right) = \frac{1+(1-\eps)\lambda -(1-\eps)\lambda}{1+(1-\eps)\lambda} = \frac{1}{1+(1-\eps)\lambda}
	\end{align*}
	The expected value of the best candidate in hindsight is: $\frac{1-\eps}{1+(1-\eps)\lambda}+1 = \frac{\lambda+2-\eps(\lambda+1)}{1+(1-\eps)\lambda}$. Thus, we have that
$	\frac{\expect[V^*]}{\expect[V_\lambda]}= \lambda+2-\eps(\lambda+1)$.
\end{proof}
\begin{claim}
	For every $\lambda>0$, there exists an instance such that the ratio $\frac{\expect[V_0]}{\expect[V_\lambda]}$ is arbitrarily close to $\lambda+1$.
\end{claim}
\begin{proof}
	Consider the family of instances defined in Example \ref{exm-ratio}. We already observed that a $\lambda$-biased agent always selects the first candidate. We now consider an unbiased agent. Notice that the expected value of the second candidate is higher than the expected value of the first candidate. Thus the unbiased agent always selects the second candidate, implying that:
		\begin{align*}
	\frac{\expect[V_0]}{\expect[V_\lambda]}= \frac{1+(1-\eps)\lambda}{1} = \lambda+1-\eps \cdot \lambda
	\end{align*}
\end{proof}

It is interesting to observe that the family of instances defined in Example \ref{exm-ratio} can also be used to show that the ratio $\frac{\expect[V_{\lambda'}]}{\expect[V_\lambda]}$ for any $0 \leq \lambda'< \lambda$ can be as high as $\lambda+1$. The reason for this is that since the $\lambda$-biased agent is indifferent between selecting the first candidate and the second candidate, any $\lambda'$-biased agent would strictly prefer selecting the second candidate.

Our main results in this section demonstrate that the lower bounds we just established for $\frac{\expect[V^*]}{\expect[V_\lambda]}$ and $\frac{\expect[V_0]}{\expect[V_\lambda]}$ are tight. The first proof builds on \cite{samuel1984comparison} showing a $2$-approximation for prophet inequalities can always be achieved by a simple threshold strategy.

\begin{theorem} \label{thm-prophet}
$\dfrac{\expect[V^*]}{\expect[V_\lambda]}\leq \lambda+2$.
\end{theorem}
\begin{proof}
	We consider a simple threshold strategy and show that even though this is not necessarily the optimal stopping rule,
	 the expected value for using this strategy is at least $\frac{\expect[V^*]}{\lambda+2}$. 

	We define our threshold strategy as follows.
	\begin{definition} [$\theta$-threshold strategy]
	The {\em $\theta$-threshold strategy} is the stopping rule that selects the first candidate $t$ such that $v_t> \theta$, where $\theta$ is defined such that $Pr(V^* > \theta) = \frac{\lambda+1}{\lambda+2} \triangleq \alpha$. \footnote{Notice that if the distribution
			of $V^*$ contains point masses then such a $\theta$ does not necessarily exist. In Appendix \ref{app-tiebreak} we show how to extend  $\theta$-threshold strategies to accommodate such cases as well.
	  }
	If there is no such candidate the strategy selects the last candidate.
	\end{definition}
Notice that the $\theta$-threshold strategy ensures us that we select a candidate of value at least $\theta$ with probability $\alpha$. We denote this strategy by $\tau_{\lambda}$ and the optimal $\lambda$-biased stopping rule by $\pi_\lambda$.

We show a stronger claim: the ratio between $\expect[V^*]$ and the \emph{utility} of a $\lambda$-biased agent using this threshold strategy is at most $\lambda+2$. This implies a bound on $\expect[V^*]/ \expect[V_\lambda]$ since
\begin{align*}
\dfrac{\expect[V^*]}{\expect[V_\lambda]} \leq
\dfrac{\expect[V^*]}{\expect[V(\pi_\lambda)] - \lambda \cdot \expect[L(\pi_\lambda)]} \leq
\dfrac{\expect[V^*]}{\expect[V(\tau_\lambda)] - \lambda \cdot \expect[L(\tau_\lambda)]} .
\end{align*}
where $V(\sigma)$ and $L(\sigma)$ are random variables equal, respectively, to the value of the selected candidate and the loss (i.e., difference between the value of the reference value and the chosen candidate, in case this difference is positive) of the stopping rule $\sigma$. The second inequality is due to the fact that the optimal $\lambda$-biased stopping rule aims to maximize the expected utility of the $\lambda$-biased agent.

We now show that $\frac{\expect[V^*]}{\expect[V(\tau_\lambda)] - \lambda \cdot \expect[L(\tau_\lambda)]} \leq \lambda +2$. We begin by bounding the expected value of the candidate selected by the threshold strategy.  Note that
\begin{align*}
\expect[V(\tau_\lambda)] = \int_0^\infty Pr(V(\tau_\lambda)>y)dy = \int_0^\theta Pr(V(\tau_\lambda)>y)dy +\int_\theta^\infty Pr(V(\tau_\lambda)>y)dy
\end{align*}
We bound each of these integrals separately. First we observe that
\begin{align*}
 \int_0^\theta Pr(V(\tau_\lambda)>y)dy \geq  \int_0^\theta \alpha \, dy = \alpha \cdot \theta
\end{align*}
This is simply because we accept a candidate that its value is higher than $\theta$ with probability $\alpha$. Thus, for any $y<\theta$ the probability that we accept a candidate with value higher than $y$ is at least $\alpha$.
For the second integral we bound the value of $Pr(V(\tau_\lambda)>y)$ as follows:
\begin{align*}
Pr(V(\tau_\lambda)>y) = \sum_{t=1}^n Pr(V_t > y, V_{-t} < \theta)
\end{align*}
where $V_{-t} < \theta$ is the event in which the values of all candidates prior to $t$ were below $\theta$. As these are two independent events we have that:
\begin{align*}
\sum_{t=1}^n Pr(V_t > y, V_{-t} < \theta) =  \sum_{t=1}^n Pr(V_t > y)\cdot Pr(V_{-t} < \theta)
\end{align*}
Observe that $Pr(V_{-t} < \theta) \geq (1-\alpha)$ since with probability $1-\alpha$ we have that in each of the time steps the value of the candidate was less than $\theta$ and this is just a smaller range. Moreover by taking a union bound we get that $ \sum_{t=1}^n Pr(V_t > y) \geq Pr(V^*>y)$, since the value of the best candidate exceeds $y$ if and only if there exists $t$ with $v_t > y$. Thus we conclude that:
\begin{align*}
\expect[V(\tau_\lambda)] \geq \alpha \cdot \theta + (1-\alpha)\int_{\theta}^{\infty} Pr(V^*>y)dy
\end{align*}
Observe that:
\begin{align*}
\int_{\theta}^{\infty} Pr(V^*>y) \, dy =
\int_{0}^{\infty} Pr(V^* - \theta > z) \, dz =
\expect[(V^*-\theta)^+]
\end{align*}

Thus, we have that:
$\expect[V(\tau_{\lambda})] \geq \alpha \cdot \theta+(1-\alpha)\expect[(V^*-\theta)^+]$. Next, we observe that the expected loss of the threshold strategy is at most $(1-\alpha)\theta$ since the threshold strategy experiences zero loss except when it selects the last candidate, an event that happens with probability $1-\alpha$ and results in loss bounded above by $\theta$ in the worst case that the last candidate has a value of at least $0$ while the maximum preceding value is $\theta-\eps$. Putting this together we get that:
\begin{align*}
\expect[V(\tau_\lambda)]-\lambda \expect[L(\tau_{\lambda})] &\geq \alpha \cdot \theta +(1-\alpha)\expect[(V^*-\theta)^+] - \lambda(1-\alpha) \cdot \theta \\
&= (\alpha-\lambda(1-\alpha) )\cdot \theta + (1-\alpha)\expect[(V^*-\theta)^+]
\end{align*}
By plugging in $\alpha =\frac{\lambda+1}{\lambda+2}$ we get that:
\begin{align*}
\expect[V(\tau_\lambda)]-\lambda \expect[(L\tau_{\lambda})] \geq \frac{1}{\lambda+2} \cdot \theta + \frac{1}{\lambda+2} \cdot \expect[(V^*-\theta)^+] \geq  \frac{1}{\lambda+2} \cdot \expect[V^*]
\end{align*}
\end{proof}

\begin{theorem} \label{thm-ops-bias-ratio}
	$\dfrac{\expect[V_0]}{\expect[V_\lambda]} \leq \lambda +1$
\end{theorem}
\begin{proof}
We denote by $\pi_\lambda$ the optimal $\lambda$-biased stopping rule and let $\Delta =\frac{\expect[V^*]}{\expect[V_0]}$ and $ \beta =\frac{\expect[V^*]}{\expect[V_\lambda]}$. We would like to bound $\frac{\beta}{\Delta} = \frac{\expect[V_0]}{\expect[V_\lambda]}$. To do so, we will get an upper bound on $\beta$ as a function of $\Delta$. We obtain this upper bound by lower bounding the utility of the $\lambda$-biased agent by using its utility from applying the optimal (unbiased) stopping rule (i.e., $\expect[V(\pi_{\lambda})]-\lambda \expect[L(\pi_\lambda)] \geq \expect[V(\pi_{0})]-\lambda \expect[L(\pi_{0})]$).  
First observe the following bound on the expected loss of using the optimal stopping rule:
\begin{align*}
\expect[L(\pi_0)] \leq \expect[V^*- V_0] = \expect[V^*]-\expect[V_0]
\end{align*}
 Notice that using our notation we have that $\expect[V_0] = \frac{\expect[V^*]}{\Delta}$. Thus, we get that $\expect[L(\pi_0)]\leq\expect[V^*] -\frac{\expect[V^*]}{\Delta} = \expect[V^*](1-\frac{1}{\Delta}) $. Hence,
\begin{align*}
\expect[V(\pi_{\lambda})]-\lambda \expect[L(\pi_\lambda)] \geq \expect[V(\pi_{0})]-\lambda \expect[L(\pi_{0})]  \geq  \frac{1}{\Delta} \expect[V^*] - \lambda \left(1-\frac{1}{\Delta} \right)\expect[V^*]
\end{align*}
On the other hand: $\expect[V(\pi_{\lambda})]-\lambda \expect[L(\pi_\lambda)] \leq \expect[V(\pi_{\lambda})] = \frac{\expect[V^*]}{\beta}$. This implies that:
\begin{align*}
\frac{1}{\beta} \geq \frac{1}{\Delta}  - \lambda \left(1-\frac{1}{\Delta} \right) \implies \beta \leq \frac{\Delta}{1-\lambda(\Delta-1)}
\end{align*}
By applying the bound on $\beta$ we proved in Theorem \ref{thm-prophet} we have that $\beta \leq \lambda+2$. Thus,
\begin{align*}
\dfrac{\expect[V_0]}{\expect[V_\lambda]} = \frac{\beta}{\Delta} \leq \min \left\{  \frac{1}{1-\lambda(\Delta-1)}, \frac{\lambda+2}{\Delta} \right\}
\end{align*}
The minimum is maximized when:
\begin{align*}
\frac{1}{1-\lambda(\Delta-1)} =  \frac{\lambda+2}{\Delta} \implies \Delta = \frac{\lambda+2}{\lambda+1}
\end{align*}
and for this value of $\Delta$ we get that $\frac{\expect[V_0]}{\expect[V_\lambda]}\leq \lambda+1$.
\end{proof}

\section{Monotonicity} \label{sec-mon} In this section we prove several monotonicity properties as the parameters of the model vary. This includes showing that the agent's expected utility monotonically decreases as the reference value increases. This enables us to show that the optimal stopping rule can be described as a threshold strategy, in which the threshold for each step is history dependent. Somewhat surprisingly the expected value of the selected candidate may actually increase when the reference value increases.

We also show that both the agent's expected utility and the expected value of the selected candidate are decreasing as $\lambda$ increases. Finally, we consider the effect of adding more candidates on the agent's expected utility and the expected value of the selected candidate. For rational agents it is clear that
additional candidates can only increase the expected value of the selected candidate. However, for biased agents we observe a  significant different effect that is based on the new candidates' location in the sequence. If the candidates are added in the beginning of the sequence then the agent's expected utility and the expected value of the selected candidate can decrease by as much as a factor of $\lambda+1$. On the other hand, adding candidates to the end of the sequence can only increase the agent's expected utility and the expected value of the candidate that will be selected.

\subsection{Monotonicity in the Reference Value and Threshold Strategies}
In Section \ref{sec-model} we mention that the optimal $\lambda$-biased stopping rule can also be described as a sequence of thresholds such that each threshold is a function of the reference value. Before presenting the proof we prove an auxiliary observation that the expected utility of the agent monotonically decreases as the reference value increases.

\begin{observation} \label{obs-mon-best}
	For $v'<v$ and any $t\geq 1$ we have that $U_\lambda[v',t] \geq U_\lambda[v,t]$.
\end{observation}
\begin{proof}
	To see why this is the case observe that an agent with reference value $v'<v$ can use the optimal $\lambda$-biased stopping rule for reference value $v$. The expected utility of the agent for using this stopping rule is at least $U_\lambda[v,t]$ since the expected value of the candidate selected will be the same and the expected loss will only be smaller.
\end{proof}


Interestingly, it is not necessarily the case that as the reference value increases the expected value of the selected candidate decreases. To see that, consider the following example: $V_1=1$ and $V_2$ is $3$ with probability $\frac12$ and $0$ with probability $\frac12$. An agent that has at the beginning a reference value of $2$ will choose candidate $2$ since $1-\lambda(2-1) < \frac32-\frac12 \cdot \lambda (2-0)$ for any $\lambda\geq 0$. On the other hand, for $\lambda>1$ an agent with reference value $0$ will choose candidate 1 since $1 > \frac32-\frac12 \lambda (2-1)$. As a result the expected value of the selected candidate will decrease from $3/2$ to $1$.

We are now ready to define the optimal $\lambda$-biased stopping rule as a threshold strategy.
\begin{proposition} \label{prop-th}
	For every $t\geq 1$ and $v\geq 0$ the optimal $\lambda$-biased stopping rule selects candidate $t$ when the reference value is $v$ if and only if $v_t \geq \theta(v,t)$ such that:
	\begin{align*}
	\theta(v,t) = \begin{cases}
	\frac{U_\lambda[v,t+1] +\lambda v}{1+\lambda}, & \text{for } v>U_\lambda[v,t+1] \\
	u, & \text{for } v\leq U_\lambda[v,t+1]
	\end{cases}
	\end{align*}
	where $u\geq v$  is the minimal value for which $u\geq U_\lambda[u,t+1]$.
\end{proposition}
\begin{proof}
	We distinguish between two cases. First we consider the case that  $v>U_\lambda[v,t+1] $. Roughly speaking, this is the case in
	which the realized values of the candidates were lower than expected, and now to maximize its expected utility the agent may need to incur some loss with respect to the reference value. As can be expected, in this case the agent always selects candidate $t$ with value $v_t\geq v$. This is because by applying Observation \ref{obs-mon-best} we have that
	$v_t\geq v>U_\lambda[v,t+1]\geq U_\lambda[v_t,t+1]$ which implies that $v_t\geq U_\lambda[v_t,t+1]$ and hence should be selected. The agent may also select a candidate such that $v_t <v$. This would be the case when:
	\begin{align*}
	v_t - \lambda(v-v_t) \geq U_\lambda[v,t+1] \implies v_t \geq \frac{U_\lambda[v,t+1] +\lambda v}{1+\lambda}
	\end{align*}
	Observe that $\frac{U_\lambda[v,t+1] +\lambda v}{1+\lambda} < v$ since $U_\lambda[v,t+1]<v$, hence we conclude that in this case candidate $t$ will be selected if and only if $v_t>\frac{U_\lambda[v,t+1] +\lambda v}{1+\lambda}$.

	We next consider the case that $v\leq U_\lambda[v,t+1] $. This is the case for the first candidate, for example, when the reference value is $0$. It is also clear that in this case the agent would never select a candidate with value strictly less than $v$.
	Hence, candidate $t$ will be selected if and only if
	\begin{align*}
	v_t \geq U_\lambda[v_t,t+1].
	\end{align*}
	Particularly, a candidate of value $v_t \geq u\geq v$ where $u$ is the minimal value such that $u\geq U_\lambda[u,t+1]$ will always be selected.
	Notice that, such a value necessarily exists since the left hand-side of the inequality $u \geq U_\lambda[u,t+1]$ is increasing with $u$ and by Observation \ref{obs-mon-best} the right hand-side is decreasing with $u$. Also, from Observation \ref{obs-mon-best} we get that the agent will select any candidate such that $v_t \geq u$ since $v_t \geq u\geq U_\lambda[u,t+1] \geq U_\lambda[v_t,t+1]$.
\end{proof}

\subsection{Monotonicity in $\lambda$}
We show that the expected utility of the agent is decreasing as $\lambda$ increases.
\begin{observation} \label{obs-mon-lambda}
	For any reference value $v\geq0, t\geq 1$ and $\lambda'>\lambda\geq 0$ we have that $U_\lambda[v,t] \geq U_{\lambda'}[v,t]$.
\end{observation}
\begin{proof}
	Recall that $U_{\lambda}[v,t]$ is the expected utility of a $\lambda$-biased agent with reference value $v$ that is now examining candidate $t$. Instead of using its optimal stopping rule the $\lambda$-biased agent can apply the optimal stopping rule for a $\lambda'$-biased agent. By doing so it will achieve utility of at least $U_{\lambda'}[v,t]$ since both agents will select the same candidate and the loss of the $\lambda$-biased agent will be smaller since $\lambda<\lambda'$.
\end{proof}

The decrease in the agent's expected utility can be merely an artifact of the greater loss that the agent exhibits as we increase the bias parameter. However, we show that an additional reason for the decrease in the utility is that a candidate with lower expected value may be selected. To prove this we show that the optimal $\lambda'$-biased stopping rule for $\lambda'>\lambda$ would select either the same candidate selected by an optimal $\lambda$-biased stopping rule or a candidate prior to it in the sequence. We formally define the following notion:
\begin{definition}
	A stopping rule $\pi$ is \emph{more patient} than a stopping rule $\pi'$ if for every realization of the candidates the stopping rule $\pi$ either selects the same candidate as $\pi'$ or selects a candidate that is later in the sequence. In other words, if the stopping rule $\pi'$ selected candidate $t'$ then $\pi$ selects a candidate $t\geq t'$.
\end{definition}
Notice that if both $\pi$ and $\pi'$ are threshold stopping rules and are defined on the same sequence of candidates, then if $\pi$ is more patient than $\pi'$ then for any $t$ and any realization $v_1,\ldots, v_{t-1}$ the threshold of $\pi$ for candidate $t$ will be greater than or equal to the threshold of $\pi'$.

The usefulness of this notion comes from the next proposition showing that, under some conditions, a stopping rule which is more patient would select candidates of higher expected values.
\begin{proposition} \label{prop-more-patient}
	If an optimal $\lambda$-biased stopping rule $\pi_{\lambda}$ is more patient than another stopping rule $\pi$ then the expected value of the candidate selected by $\pi_\lambda$ is higher than the expected value of the candidate selected by $\pi$ (i.e., $\expect[V(\pi_\lambda)]\geq \expect[V(\pi)]$).
\end{proposition}
\begin{proof}
	Since $\pi_{\lambda}$ is more patient, then for every realization $v_1,\ldots,v_t$ such that $\pi$ selects candidate $t$,  $\pi_{\lambda}$ selects candidate $t$ or a candidate that will arrive later. Let $V(\pi_\lambda|v_1,\ldots,v_t)$ and $L(\pi_\lambda|v_1,\ldots,v_t)$ be two random variables that are equal respectively to the value and loss of the optimal $\lambda$-biased stopping rule for the realization $v_1,\ldots,v_t$.
	Since $\pi_\lambda$ is optimal for a $\lambda$-biased agent, we have that:
		\begin{align} \label{eq-opt}
	\expect[V(\pi_\lambda|v_1,\ldots,v_t)] - \lambda \expect[L(\pi_\lambda|v_1,\ldots,v_t)] \geq v_t - \lambda(\max\{v_1,\ldots,v_{t-1}\}-v_t)^+
	\end{align}
	We claim that for every realization $v_1,\ldots,v_t$ it has to be the case that $\expect[V(\pi_\lambda|v_1,\ldots,v_t)] \geq v_t$. This will in turn imply the proposition. Assume towards contradiction that there exists some realization for which  $\expect[V(\pi_\lambda|v_1,\ldots,v_t)] < v_t$. By Equation (\ref{eq-opt}), this implies that $\expect[L(\pi_\lambda|v_1,\ldots,v_t)] < (\max\{v_1,\ldots,v_{t-1}\}-v_t)^+$. The expected loss is determined by the difference between the reference value and the value of the selected candidate. Since $\pi_{\lambda}$ is more patient than $\pi$  the reference value at the time of the selection may only increase. Thus, if the expected value of the selected candidate decreases the expected loss would increase.
%
%
\end{proof}

Next, we show that as the value of $\lambda$ decreases the optimal $\lambda$-biased stopping rule becomes more patient.

\begin{claim}
The optimal $\lambda$-biased stopping rule is more patient than the optimal $\lambda'$-biased stopping rule for $\lambda \leq \lambda'$.
\end{claim}
\begin{proof}
	Assume towards contradiction that there exists some realization $v_1,\ldots, v_{t}$ such that $\pi_{\lambda}$ selects candidate $t$ and $\pi_{\lambda'}$ selects candidate $t'$ such that $t'>t$. We show that in this case a $\lambda'$-biased agent can increase its expected utility by instead selecting candidate $t$ on any such realization.
	Let $v = \max \{v_1,\ldots,v_{t-1}\}$ denote the reference value. 
	As in the proof of Proposition \ref{prop-more-patient}, let $V(\pi_{\lambda'}|v_1,\ldots,v_t)$ and $L(\pi_{\lambda'}|v_1,\ldots,v_t)$ be two random variables that are equal respectively to the value and loss of the stopping rule $\pi_{\lambda'}$ for the realization $v_1,\ldots,v_t$.
	Since $\pi_{\lambda}$ is optimal for a $\lambda$-biased agent we have that the utility for following it is greater than using $\pi_{\lambda'}$ instead:
	\begin{align} \label{eq-patient}
	v_t - \lambda(v-v_t)^+ \geq \expect[V(\pi_{\lambda'}|v_1,\ldots,v_t)] - \lambda \expect[L(\pi_{\lambda'}|v_1,\ldots,v_t)]
	\end{align}

	We show that $(v-v_t)^+ \leq \expect[L(\pi_{\lambda'}|v_1,\ldots,v_t)]$. By Inequality (\ref{eq-patient}) we have that this has to be the case if $v_t \leq\expect[V(\pi_{\lambda'}|v_1,\ldots,v_t)]$. Hence, we are left with handling the case that $v_t > \expect[{V}(\pi_{\lambda'}|v_1,\ldots,v_t)]$. Now, observe that $\expect[L(\pi_{\lambda'}|v_1,\ldots,v_t)] \geq v- \expect[V(\pi_{\lambda'}|v_1,\ldots,v_t)]$. The reason for this is twofold: first the reference value may increase above $v$ and second in the expression $v- \expect[V(\pi_{\lambda'}|v_1,\ldots,v_t)]$, all the events in which the selected candidate has a higher value than $v$ will contribute a negative number. Putting this together with the assumption that $v_t > \expect[V(\pi_{\lambda'}|v_1,\ldots,v_t)]$ we get that $\expect[L(\pi_{\lambda'}|v_1,\ldots,v_t)] \geq v- \expect[V(\pi_{\lambda'}|v_1,\ldots,v_t)] \geq v-v_t$. Since $\expect[L(\pi_{\lambda'}|v_1,\ldots,v_t)]\geq 0$ we get that  $\expect[L(\pi_{\lambda'}|v_1,\ldots,v_t)]\geq (v-v_t)^+$ in this case as well.

	To complete the proof, we add the inequality  $-(\lambda'-\lambda)(v-v_t)^+ \geq -(\lambda'-\lambda)\expect[L(\pi_{\lambda'}|v_1,\ldots,v_t)]$ to Inequality (\ref{eq-patient}) and get
	\begin{align*}
	v_t - \lambda'(v-v_t)^+ \geq \expect[{V}(\pi_{\lambda'}|v_1,\ldots,v_t)] - \lambda' \cdot \expect[L(\pi_{\lambda'}|v_1,\ldots,v_t)].
	\end{align*}
	A contradiction is reached since this implies that $\pi_{\lambda'}$ should also select candidate $t$ as well.
\end{proof}

By applying Proposition \ref{prop-more-patient} we reach the following corollary:
\begin{corollary}
	Let $\lambda'\geq \lambda$. The expected value of the candidate selected by a $\lambda$-biased agent is higher than the expected value of a candidate selected by a $\lambda'$-biased agent,
	i.e.~$\expect[V(\pi_{\lambda})]\geq\expect[V(\pi_{\lambda'})]$.
\end{corollary}
\subsection{Monotonicity and Non-Monotonicity in the Number of Candidates}
Interestingly, monotonicity in the number of candidates --- which is trivial for unbiased agents (i.e., free disposal) ----- does not necessary hold for biased agents with $\lambda>0$. In particular, we observe that while adding more candidates to the end of the sequence can only increase the expected utility (the argument is similar to unbiased agents), adding more candidates to the beginning of the sequence can actually reduce the expected utility. We begin by showing that adding candidates at the end of the sequence can only increase the expected utility of the agent.

In this section, because we are contemplating varying the sequence of candidates, we modify our notation for optimal $\lambda$-biased stopping rules to explicitly indicate the sequence for which the stopping rule is optimized. Thus, $\pi_\lambda(V_1,\ldots, V_n)$ denotes the optimal $\lambda$-biased stopping rule for the sequence of candidates
$V_1,\ldots,V_n$, and
$U_\lambda(\pi_\lambda(V_1,\ldots, V_n))$ denotes the expected utility of a $\lambda$-biased agent using it. Similar to before, $V(\pi_\lambda(V_1,\ldots, V_n))$ is a random variable that equals the expected value of the candidate selected by $\pi_\lambda(V_1,\ldots, V_n)$.

\begin{claim} \label{clm-mon-more-end}
For any $\lambda\geq 0$ we have that $U_\lambda(\pi_\lambda(V_1,\ldots,V_{n+1})) \geq U_\lambda(\pi_\lambda(V_1,\ldots,V_n))$ and
\\ ${\expect[V(\pi_\lambda(V_1,\ldots, V_n,V_{n+1}))] \geq \expect[V(\pi_\lambda(V_1,\ldots, V_n))]}$.
\end{claim}
\begin{proof}
Observe that a $\lambda$-biased agent facing the sequence $V_1,\ldots,V_n,V_{n+1}$ can use a stopping rule identical to the one an agent facing $V_1,\ldots,V_n$ uses. By doing so its expected utility will be identical to the expected utility of the agent facing $V_1,\ldots,V_n$. Thus,  $U_\lambda(\pi_\lambda(V_1,\ldots,V_{n+1})) \geq U_\lambda(\pi_\lambda(V_1,\ldots,V_n))$. To see why it is also the case that $\expect[V(\pi_\lambda(V_1,\ldots, V_n,V_{n+1}))] \geq \expect[V(\pi_\lambda(V_1,\ldots, V_n))]$ we observe that the optimal $\lambda$-biased stopping rule for  $V_1,\ldots,V_{n+1}$ is more patient than the optimal $\lambda$-biased stopping rule for  $V_1,\ldots,V_{n}$.\footnote{To be completely formal about this, notice that one can trivially extend $\pi_\lambda(V_1,\ldots,V_n)$ to be defined over $V_1,\ldots V_{n+1}$.} To see why this is the case, assume towards contradiction that there exists a
realization $v_1,\ldots,v_t$ such that $\pi_\lambda(V_1,\ldots, V_n,V_{n+1})$ selects candidate $t$ and $\pi_\lambda(V_1,\ldots, V_n)$ selects candidate $t'>t$. Notice that since both stopping rules are optimal for a $\lambda$-biased agent, the utility in each of them for selecting candidate $t$ is the same. Now, if the stopping rule $\pi_\lambda(V_1,\ldots, V_n)$ selects candidate $t'>t$ it means that the expected utility for doing so is greater than the utility of selecting candidate $t$. We reach a contradiction, since $\pi_\lambda(V_1,\ldots, V_n,V_{n+1})$ can guarantee a utility at least as high as the utility of $\pi_\lambda(V_1,\ldots, V_n)$ which means it could not have selected candidate $t$. By Proposition \ref{prop-more-patient} this implies that $\expect[V(\pi_\lambda(V_1,\ldots, V_n,V_{n+1}))] \geq \expect[V(\pi_\lambda(V_1,\ldots, V_n))]$ as required.

\end{proof}

We now consider adding a candidate at the beginning of the sequence. Such addition can decrease the expected utility of a $\lambda$-biased agent. Example \ref{exm-ratio} illustrates this: If the agent only considers the candidate with the expected value of $1$, then the agent's expected utility is $1$. However, if the candidate with the fixed value of $\frac{1}{1+(1-\eps)\lambda}$ is added to the beginning of the sequence, then the expected utility of the $\lambda$-biased agent decreases to $\frac{1}{1+(1-\eps)\lambda}$. Thus, the expected utility decreased by a factor of $\lambda+1$. We show that this ratio is tight:
\begin{claim}
	For any $\lambda\geq 0$ we have that $U_\lambda(\pi_\lambda(V_0,V_1,\ldots,V_{n})) \geq \frac{1}{\lambda+1} U_\lambda(\pi_\lambda(V_1,\ldots,V_n))$ and
	\\ ${\expect[V(\pi_\lambda(V_0,V_1,\ldots, V_n))] \geq \frac{1}{\lambda+1}  \expect[V(\pi_\lambda(V_1,\ldots, V_n))]}$.
\end{claim}
\begin{proof}
We lower bound $U_\lambda(\pi_\lambda(V_0,V_1,\ldots,V_{n}))$ by considering two alternative stopping rules that a $\lambda$-biased agent may use in this instance:
\begin{enumerate}
	\item Always select candidate $0$ - the expected utility is $\expect[V_0]$.
	\item Never select candidate $0$ and then continue as in the optimal stopping rule for $V_1,\ldots,V_n$  - the expected utility is at least $U_\lambda(\pi_\lambda(V_1,\ldots,V_n))-\lambda \expect[V_0]$.
\end{enumerate}
Since, each of these stopping rules is valid we get that:
\begin{align*}
U_\lambda(\pi_\lambda(V_0,V_1,\ldots,V_{n})) \geq \max \{\expect[V_0], U_\lambda(\pi_\lambda(V_1,\ldots,V_n))-\lambda \expect[V_0] \}
\end{align*}
Notice that if $\expect[V_0] \geq U_\lambda(\pi_\lambda(V_1,\ldots,V_n))-\lambda \expect[V_0]$ we get that $\expect[V_0] \geq \frac{U_\lambda(\pi_\lambda(V_1,\ldots,V_n))}{\lambda+1}$, hence $U_\lambda(\pi_\lambda(V_0,V_1,\ldots,V_{n})) \geq \frac{U_\lambda(\pi_\lambda(V_1,\ldots,V_n))}{\lambda+1}$. Else,  $\expect[V_0] < \frac{U_\lambda(V_1,\ldots,V_n)}{\lambda+1}$
and in this case,we get that 
\begin{align*}
U_\lambda(\pi_\lambda(V_0,V_1,\ldots,V_{n})) \geq U_\lambda(V_1,\ldots,V_{n}) -\lambda \expect[V_0] &\geq U_\lambda(\pi_\lambda(V_1,\ldots,V_{n})) -\lambda \frac{U_\lambda(\pi_\lambda(V_1,\ldots,V_n))}{\lambda+1} \\
&=  \frac{U_\lambda(\pi_\lambda(V_1,\ldots,V_{n}))}{\lambda+1} 
\end{align*}
 as required.
To observe that ${\expect[V(\pi_\lambda(V_0,V_1,\ldots, V_n))] \geq \frac{1}{\lambda+1}  \expect[V(\pi_\lambda(V_1,\ldots, V_n))]}$ notice that adding a candidate can never decrease the expected value of the candidate selected by an unbiased agent. Thus, we have that ${\expect[V(\pi_0(V_0,V_1,\ldots, V_n))] \geq \expect[V(\pi_0(V_1,\ldots, V_n))]}$. By applying Theorem \ref{thm-ops-bias-ratio} we get that ${\expect[V(\pi_\lambda(V_0,V_1,\ldots, V_n))] \geq \frac{1}{\lambda+1}   \expect[V(\pi_0(V_0,V_1,\ldots, V_n))]}$. Putting these together we get that ${\expect[V(\pi_\lambda(V_0,V_1,\ldots, V_n))] \geq \frac{1}{\lambda+1}  \expect[V(\pi_\lambda(V_1,\ldots, V_n))]}$ as required.
\end{proof}

\section{Ordering Problems} \label{sec-var} In the classic prophet inequality setting we assume that the order of the candidates is predetermined. In this section we relax the assumption. We first consider the setting in which the order of the candidates is chosen uniformly at random. We show two types of bounds here with respect to the prophet: a tight and easy bound of $n$ and a much more challenging bound of about $\ln(\lambda)$. Both bounds provide a significant improvement over the bounds for fixed order, and both are tight (up to constant factors) as a function of their respective parameters.

In the second ordering problem we consider, the agent can decide on the ordering of the candidates. We assume that the agent will do so in order to maximize the expected value of the selected candidate. This would be, for example, the objective of an agent who is sophisticated about its bias (i.e., is aware of it) and hence when ordering the candidates {it tries to maximize} the expected \emph{value} of the selected candidate. We mainly focus on distributions whose support only includes $2$ points (henceforth, $2$-point distributions). Such distributions were studied --- in the ordinary setting of (unbiased) optimal stopping --- by \cite{agrawal2020optimal}, which showed a bound of $5/4$ on the prophet inequality. For the optimal $\lambda$-biased stopping rule we show a bound of $2$, irrespective of the value of $\lambda$, which is a significant improvement with respect to the bound of $n$ that we have from Claim \ref{clm-bound-n} for $3$-point distributions. 

\subsection{Random Ordering}
In this section we assume that the order of the candidates is determined uniformly at random. In such a setting it is not hard to see that the expected value of the selected candidate is at least $\expect[V^*]/n$. To see why this is the case, consider the stopping rule that always selects the first candidate. The expected value of the candidate selected by this rule is $\sum_{i=1}^n \frac{1}{n} \expect[V_i] \geq \frac{1}{n} \expect[V^*]$. For this stopping rule the expected utility of the agent equals the expected value of the selected candidate. Thus, the expected value of the candidate selected by the optimal $\lambda$-biased stopping rule can only be higher. This proves the following claim.

\begin{claim} \label{clm-bound-n}
	The ratio between the prophet and a biased agent for random order is at most $n$.
\end{claim}

We demonstrate the tightness of this bound by considering a family of instances in which the values of all the candidates are drawn from the same distribution:
\begin{claim} \label{clm-iid-n}
	There exists a family of probability distributions parameterized by $n$ such that as $n$ approaches infinity and $\lambda$ is sufficiently large (as a function of $n$) $\expect[V^*] / \expect[V(\pi_\lambda)]$ approaches $n$.
\end{claim}
\begin{proof}
	Consider the distribution:
$	V=
	\begin{cases}
	0, ~~~ w.p~ \frac{1}{n} \\
	\frac{1}{n^3}, ~~ w.p~ 1-\frac{1}{n^2}-\frac{1}{n}\\
	1, ~~~ w.p~ \frac{1}{n^2}
	\end{cases} .
	$
\\
	We bound the expected value of the best candidate using the probability that its value is $1$:
	\begin{equation} \label{eq:iid-n.1}
	\expect[V^*] \geq 1\cdot \left( 1- \left( 1-\frac{1}{n^2} \right)^n \right)
	\end{equation}
	As for the biased agent, we pick $\lambda$ large enough that the optimal $\lambda$-biased stopping rule always selects the first candidate that has non-zero value. Denote by $OPT(\lambda,n)$ the expected value of this stopping rule for $n$ candidates. To solve for the expected value of the candidate selected by the $\lambda$-biased agent we let $OPT(\lambda,n) = \expect[V(\pi_\lambda(\underbrace{V,\ldots,V}_n))]$ and define and solve the following recurrence relation:
	\begin{align*}
	OPT(\lambda,n) = \expect[V] + \frac{1}{n} \cdot OPT(\lambda,n-1)
	\end{align*}
	where $OPT(\lambda,0) = 0$. Thus, we have that:
	\begin{equation} \label{eq:iid-n.2}
	\expect[V(\pi_\lambda)] = OPT(\lambda,n) =\sum_{i=0}^{n-1} n^{-i} \cdot \expect[V]  = \expect[V] \cdot \frac{1-n^{-n}}{1-n^{-1}} .
	\end{equation}
	Combining equations~\eqref{eq:iid-n.1} and~\eqref{eq:iid-n.2} we get that
	\begin{equation} \label{eq:iid-n.3}
	\frac{\expect[V^*]}{\expect[V(\pi_\lambda)]} \geq \frac{1- \left( 1-\frac{1}{n^2} \right)^n}{\expect[V] \cdot \frac{1-n^{-n}}{1-n^{-1}} } \geq
	\frac{ 1- \left( 1  -\frac{1}{n^2} \right)^n }{ \left( \frac{1}{n^2}+\frac{1}{n^3} \right) \cdot \frac{1}{1-n^{-1}} }
	 = (n-1) \cdot \frac{1-(1-\frac{1}{n^2})^n}{(\frac{1}{n^2}+\frac{1}{n^3}) \cdot n} .
 	\end{equation}
	Observe that
	\begin{align*}
	{\left(1- \frac{1}{n^2} \right)^n} = \sum_{i=0}^n (-1)^i \binom{n}{i} \left( \frac{1}{n^2} \right)^i 	\leq 1-\frac{1}{n}+\frac{1}{2n^2} .
	\end{align*}
	Thus, we have that
	\begin{align*}
	\frac{1- \left( 1-\frac{1}{n^2} \right)^n}{ \left( \frac{1}{n^2}+\frac{1}{n^3} \right) \cdot n} \geq \frac{\frac{1}{n}-\frac{1}{2n^2}}{\frac{1}{n}+\frac{1}{n^2} } = \frac{2n-1}{2n+2} = 1 - \frac{3}{2n+2}
	\end{align*}
	which implies that
	\begin{align*}
	\lim_{n\rightarrow \infty} \frac{1- \left( 1-\frac{1}{n^2} \right)^n}{ \left( \frac{1}{n^2}+\frac{1}{n^3} \right) \cdot n}  {\ge  \lim_{n\rightarrow \infty}	1 - \frac{3}{2n+2} = 1}.
	\end{align*}
	In combination with~\eqref{eq:iid-n.3} this implies $\frac{\expect[V^*]}{\expect[V(\pi_\lambda)]} \geq
	n - o(n)$ as $n \to \infty$. Since we established an upper bound of $n$ in Claim \ref{clm-bound-n} we conclude that the ratio $\expect[V^*] / \expect[V(\pi_\lambda)]$ approaches $n$ as required.

		For completeness we give a bound on the value of $\lambda$ for which the agent will always select a non-zero candidate. In particular, we should pick $\lambda$ such that the optimal $\lambda$-biased stopping rule selects a candidate of value $\frac{1}{n^3}$ if this is the first candidate with positive value. We need to show that for every such candidate, the potential expected loss from not taking the candidate is larger than the potential benefit of not taking it and getting a candidate of value $1$.  Observe that in any step $t$ the probability that all the next candidates will have a value of $0$ is at least $(\frac{1}{n})^{n-1}$ and in this case the agent will exhibit a loss of $\lambda \cdot \frac{1}{n^3}$. On, the other hand, the benefit from not taking the candidate can be crudely bounded by $1-\frac{1}{n^3}$. Thus, we get that the agent will follow the desired stopping rule whenever:
	\begin{align*}
	\left( \frac{1}{n} \right)^{n-1} \cdot \lambda \cdot \frac{1}{n^3} \geq 1-\frac{1}{n^3}
	\end{align*}
	which clearly holds for $\lambda = n^{n+2}$.
\end{proof}

In the previous construction we established a tight bound of $n - o(n)$ on the ratio between the prophet and a $\lambda$-biased agent by taking $\lambda$ which is exponential in $n$. We now show that this dependency is required by showing that as a function of $\lambda$ the bound between the prophet and a $\lambda$-biased agent is approximately bounded by $\ln(\lambda)$.

\begin{proposition} \label{random-order-prophet}
	For a random order the ratio
$	\dfrac{\expect[V^*]}{\expect[V(\pi_\lambda)]} \leq \rho$
	where $\rho > 1$ is the solution to the equation
	\begin{equation} \label{eq:rho}  \rho - \tfrac{1}{\lambda+1}(\rho-1) =
	\ln(\lambda+1) - \ln(1 - \rho^{-1}) .
	\end{equation}
\end{proposition}
Observe that for $\lambda=0$ we get that $\rho = e/(e-1)$.
In Appendix~\ref{app-order} we prove that
$|\rho - \ln(\lambda)| \to 0$ as $\lambda \to \infty$.




\begin{proof}
	Choose threshold $\theta$ such that
	$$\Pr(V^* < \theta) = \frac{1 - \rho^{-1}}{\lambda+1}  . $$
	As in the proof of Theorem~\ref{thm-prophet}, we denote by $\tau_\lambda$ the threshold strategy in which we select the first candidate that has value greater than $\theta$.\footnote{As in the proof of Theorem~\ref{thm-prophet}, there are technicalities that arise if the distribution of
	$V^*$ has a point-mass at $\theta$. Appendix~\ref{app-tiebreak} explains how to
	address the technicalities.} If there is no such candidate the threshold strategy selects the last candidate.

	We begin by
	observing that
	\begin{equation} \label{eq:rop.1}
	\dfrac{\expect[V^*]}{\expect[V(\pi_\lambda)]} \leq
	\dfrac{\expect[V^*]}{\expect[V(\pi_\lambda)] - \lambda \cdot \expect[L(\pi_\lambda)]} \leq
	\dfrac{\expect[V^*]}{\expect[V(\tau_\lambda)] - \lambda \cdot \expect[L(\tau_\lambda)]} .
	\end{equation}
	Let $q = \frac{1 - \rho^{-1}}{\lambda+1}$ denote the probability that
	$V^* < \theta$. To bound the numerator on the right side of
	\eqref{eq:rop.1}, note that
	\begin{equation} \label{eq:rop.2}
	 \expect[V^*] \leq \theta + \expect[(V^*-\theta)^+] \leq
	\theta + \sum_{i=1}^n \expect[(V_i-\theta)^+] .
	\end{equation}
	To bound the denominator, for each $i$ let $c_i$ denote the probability that all of the values observed
	before the arrival of $v_i$ are less than $\theta$. Then
	\begin{align}
	\label{eq:rop.3}
	\expect[V(\tau_\lambda)] & \geq (1-q) \theta + \sum_{i=1}^n c_i \expect[(V_i - \theta)^+]
	\\
	\label{eq:rop.4}
	\expect[L(\tau_\lambda)] & \le q \theta
	\\
	\label{eq:rop.5}
	\expect[V(\tau_\lambda)] - \lambda \cdot \expect[L(\tau_\lambda)] & \ge
	(1 - (\lambda+1) q) \theta + \sum_{i=1}^n c_i \expect[(V_i - \theta)^+] .
	\end{align}
	By our choice of $q$, we have $1 - (\lambda+1)q = \rho^{-1}$.
	Below we will show that for every candidate $i$, $c_i \ge \rho^{-1}$. Notice that $i$ is the index of the candidate and not its location in an ordering (the ordering is chosen uniformly at random).
	Assuming this inequality for the moment, it implies
	\begin{equation} \label{eq:rop.6}
	\expect[V(\tau_\lambda)] - \lambda \cdot \expect[(L\tau_\lambda)]  \ge
	\rho^{-1} \theta + \rho^{-1} \sum_{i=1}^n \expect[(v_i - \theta)^+]
	\end{equation}
	and, in conjunction with~\eqref{eq:rop.2}, this implies
	\begin{equation} \label{eq:rop.7}
	\dfrac{\expect[V^*]}{\expect[V(\tau_\lambda)] - \lambda \cdot \expect[L(\tau_\lambda)]} \le \rho
	\end{equation}
	which will finish the proof of the proposition.

	To bound $c_i$ from below, it will help to analyze the
	following sampling process for generating the random
	permutation of the items. We sample i.i.d.\ uniformly-distributed
	values $\alpha_1,\alpha_2,\ldots,\alpha_n$ from the interval $[0,1]$ and
	assume the items arrive in order of increasing $\alpha_i$.
	Recall that $c_i$ is the probability of the event $\mathcal{E}_i$,
	that no items arriving before $i$ have values above $\theta$.
	We have
	\begin{equation} \label{eq:rop.8}
	\mathcal{E}_i = \bigcap_{j \neq i} \mathcal{E}_{ij},
	\end{equation}
	where $\mathcal{E}_{ij}$ denotes the event that {\em either}
	$V_j < \theta$ or $\alpha_j > \alpha_i$. Let $q_j = \Pr(V_j < \theta)$.
	We have
	\begin{equation} \label{eq:rop.9}
	\Pr(\mathcal{E}_{ij} \mid \alpha_i) = (1 - \alpha_i) + \alpha_i q_j
	\end{equation}
	and the events $\{ \mathcal{E}_{ij} \mid j \neq i \}$ are
	conditionally independent, given $\alpha_i$. Hence,
	\begin{equation} \label{eq:rop.10}
	\Pr(\mathcal{E}_i \mid \alpha_i) = \prod_{j \neq i} (1 - \alpha_i + \alpha_i q_j) .
	\end{equation}
	Each factor in the product on the right side can be bounded
	from below using Jensen's Inequality:
	\begin{equation} \label{eq:rop.11}
	1 - \alpha_i + \alpha_i q_j = (1-\alpha_i) \exp(0) + \alpha_i \exp(\ln(q_j)) \geq
	\exp( \alpha_i \ln(q_j) ) .
	\end{equation}
	Multiplying these lower bounds together, we obtain
	\begin{equation} \label{eq:rop.12}
	\Pr(\mathcal{E}_i \mid \alpha_i)  \geq \exp \left( \alpha_i \sum_{j \neq i} \ln(q_j) \right).
	\end{equation}
	Now, observe that
	\begin{equation} \label{eq:rop.13}
	\sum_{j \neq i} \ln(q_j) = \ln \Pr( \max_{j \neq i} V_j \, < \, \theta )
	\geq \ln \Pr(V^* < \theta) = \ln(q)
	\end{equation}
	hence
	\begin{equation} \label{eq:rop.14}
	\Pr(\mathcal{E}_i \mid \alpha_i) \geq \exp ( \alpha_i \ln(q) ) .
	\end{equation}
	Now integrate with respect to $\alpha_i$ to obtain the unconditional
	probability:
	\begin{equation} \label{eq:rop.15}
	c_i = \Pr(\mathcal{E}_i) \geq \int_0^1 \exp(\alpha_i \ln(q) ) \, d \alpha_i
	=
	\frac{1}{\ln(q)} \left[ \exp(\ln(q)) - \exp(0) \right] = \frac{1-q}{\ln(1/q)} .
	\end{equation}
	Recalling that $q = \frac{1 - \rho^{-1}}{\lambda+1}$ and that
	$\rho - \tfrac{1}{\lambda+1}(\rho-1) =
	\ln(\lambda+1) - \ln(1 - \rho^{-1})$
	we find that
	\begin{align*}
	1-q & = \frac{\lambda + \rho^{-1}}{\lambda+1} \\
	\ln(1/q) &= \ln(\lambda+1) - \ln(1 - \rho^{-1}) = \rho - \tfrac{1}{\lambda+1}(\rho-1) \\
	\frac{1-q}{\ln(1/q)} & = \frac{\lambda + \rho^{-1}}{(\lambda+1) \rho - (\rho - 1)} =
	\frac{\lambda + \rho^{-1}}{\lambda \rho + 1} = \rho^{-1}
	\end{align*}
	which concludes the proof that $c_i \geq \rho^{-1}$, as desired.
\end{proof}

\subsection{Picking the Best Ordering}
In this section we consider the setting in which the agent, or someone else on his behalf, can choose the ordering to maximize the expected value of the selected candidate. It is not hard to see that by placing the candidate with highest expectation first, the value of the candidate selected by a $\lambda$-biased agent for any $\lambda>0$ is at least $1/n$ of the expected value of the candidate selected by a prophet. Claim \ref{clm-iid-n} which considers i.i.d.\  distributions applies to this case as well, and implies that this is tight if we allow $\lambda$ to depend exponentially on $n$. The probability distribution used in Claim \ref{clm-iid-n} has $3$ points in its support. In this section we show that this is a necessary condition by proving that for any $\lambda>0$, the prophet inequality for the best ordering when the candidates are drawn from distributions with $2$ points in their support is at most $2$. 

 Formally,  2-point distributions are defined as follows:
\begin{align*}
V_i=
\begin{cases}
h_i,~~~~w.p~~ p_i \\
l_i,~~~~w.p~~ 1-p_i
\end{cases}
\end{align*}

As in previous settings we assume that the agent knows the distributions that the candidates are drawn from.
The proof of the following proposition has a similar structure to the proof of Agrawal et al.\ \cite{agrawal2020optimal} showing that for a rational agent the ratio between the candidate selected by the prophet and by the optimal stopping rule is at most $5/4$.

\begin{proposition} \label{prop-2point-2}
	If all candidates are drawn from 2-point distributions then for every $\lambda$, there exists an ordering such that the ratio between the expected value of the candidate selected by the optimal $\lambda$-biased stopping rule and a prophet is at most $2$. This is tight as $\lambda$ approaches infinity.
\end{proposition}
\begin{proof}
	Let $X^\#$ denote a candidate with a maximal low value (i.e., $l^\# \in \arg\max{l_i}$).
	We show that one of the following two orderings guarantees that for any $\lambda>0$ the optimal $\lambda$-biased stopping rule selects a candidate of expected value at least $\expect[V^*]/2$:
	\begin{enumerate}
		\item Order all the candidates in decreasing order of $h_i$. \label{enum-decreasing}
		\item First order all candidates such that $h_i\geq E[X^\#]$ in decreasing order of high value. Then, locate $X^\#$ and after it the rest of the candidates in decreasing order of high value.	\label{enum-middle}
	\end{enumerate}
In the following analysis instead of analyzing the expected value of ordering \ref{enum-middle} we analyze the expected value of the candidate selected when ordering \ref{enum-middle} is truncated by removing all of the candidates after $X^\#$. By Claim \ref{clm-mon-more-end} we have that the expected value of the candidate selected from
ordering \ref{enum-middle} is at least as high as the expected value of the candidate selected from this
truncated ordering, which we refer to as ordering \ref{enum-middle}' henceforth.

Observe that for any $\lambda>0$, any candidate $V_i$ before $X^\#$ such that $h_i\geq \expect[X^\#]$ will be selected if and only if its high value is realized. Candidates before $X^\#$ are never selected when their low value is realized because for any such candidate we have that $l_i \leq l^\#$, so the expected utility for continuing is at least as great as the expected value for stopping, in the event that
	$l_i$ is realized. On the other hand, if $h_i$ is realized, in the case of ordering 1 it is easy to see that the optimal $\lambda$-biased rule must stop on $V_i$ because $h_i$ lies above the support of the value distribution of every remaining unobserved candidate. In the case of ordering 2', to show that the optimal $\lambda$-biased rule stops when $h_i \ge \expect[X^\#]$ is realized, we observe that if the stopping rule were to continue past candidate $i$, then conditional on stopping before $X^\#$ it must select a candidate of value at most $h_i$, and conditional on stopping at $X^\#$ it selects a candidate of expected value $\expect[X^\#]$, so in both cases the expected {utility} of continuing is no greater than the {utility} of stopping on $h_i$.
This together with the fact that in both orderings any candidate $i$ such that $h_i \geq h^\#$ is located before $X^\#$ implies that in the scenario that $V^*>h^\#$ the biased agent would select the best candidate in both orderings.

From now on we focus on the scenario that $V^*\leq h^\#$. Let $S$ denote the set of candidates whose high value is less than or equal to $h^\#$ excluding $X^{\#}$. 
We denote by $p_S$ the probability that a high value was realized for at least one of the candidates in $S$ (i.e., $p_S = 1 - \prod_{j\in S} (1-p_j)$). We also denote by $w_S$ the expected value of the best candidate among $S$ in this case (i.e., $w_S = \expect[\max_{j\in S} V_j | \exists j\in S \text{~s.t~} V_j =h_j]$ ). With this notation, the expected value of the best candidate for this scenario is:
\begin{align*}
\expect[V^*|V^*\leq h^\#] = p^\#\cdot h^\#+(1-p^\#) p_S\cdot w_S + (1-p^\#)(1- p_S) \cdot l^\#
\end{align*}

Observe that the {conditional} expected value of ordering \ref{enum-decreasing} {given $V^* \leq h^\#$} is at least $\expect[X^\#]$. Also notice that since we are interested in a bound for any $\lambda>0$ we cannot guarantee a better lower bound. The reason for this is that since $l^\#$ is the maximal low value there exists a value of $\lambda$ for which a $\lambda$-biased agent would select $X^\#$ even when a low value is realized. As for ordering \ref{enum-middle}', recall that we select any candidate $i$ prior to $X^\#$ if and only if a high value is realized
and $h_i\geq \expect[X^\#]$. Moreover, since the candidates are ordered in decreasing order of high value the agent will select the candidate
that has the highest such value. If no such value is realized the expected value of the selected candidate will be $\expect[X^\#]$.
{Letting $S'=\{i| \expect[X^\#] \leq h_i \leq h^\#\}$}, the {conditional} expected value of the selected candidate in this ordering is at least:

\begin{align*}%
	\expect \left[ \max \left\{ \max_{i\in S':V_i=h_i} h_i  ,\expect[X^\#] \right\}
	                 \right]
	& =
  \expect \left[ \max \left\{ \max_{i\in S:V_i=h_i} h_i  ,\expect[X^\#] \right\} \right]
	& \geq p_S\cdot w_S + (1-p_S)\expect[X^\#]
\end{align*}
To complete the proof we show that the sum of {conditional} expected values of {orderings 1 and 2'} is greater than {or equal to} $\expect[V^*|V^*\leq h^\#]$. To this end, observe that
\begin{align*}
\underbrace{\expect[X^\#]}_{\text{ordering 1}} + \underbrace{p_S\cdot w_S + (1-p_S)\expect[X^\#]}_{\text{ordering 2}}
& \geq p^\#\cdot h^\#+(1-p^\#) p_S\cdot w_S + (1-p^\#)(1- p_S) \cdot l^\# \\
& = \expect[V^*|V^*\leq h^\#]
\end{align*}
This implies that at least one of the orderings should provide a $2$-approximation. {Together with the fact that the expected value of the candidate selected in ordering 2 is at least as high the expected value of the candidate selected in ordering 2' this concludes the proof. }
\end{proof}

In Claim \ref{clm-2-point-tight} in the appendix we show that this bound is essentially tight. We do so by  constructing a family of instances with all candidates drawn from 2-point distributions, such that as $\lambda$ goes to infinity the ratio between the expected value of the selected candidate in the best ordering and the value of the best candidate in hindsight approaches $2$.

\bibliographystyle{plain}
\bibliography{ref}

\appendix
\section{Using Randomized Tie Breaking to Deal With Point Masses} \label{app-tiebreak}
In Theorem~\ref{thm-prophet} and Proposition~\ref{random-order-prophet}
we analyzed stopping rules defined by setting a threshold $\theta$ so as
to equate $\Pr(V^* \le \theta)$ with a specified constant. When the
distribution of $V^*$ contains no point masses, its cumulative
distribution function is continuous, so the intermediate value
theorem guarantees the existence of such a $\theta$.

However, when the distribution of $V^*$ contains point masses,
it is possible that there is no $\theta$ that makes
$\Pr(V^* \le \theta)$ precisely equal to the specified
constant. If so, one should interpret the definition of
the $\theta$-threshold strategy to include
randomized tie-breaking.

\begin{definition}[randomized $\theta$-threshold strategy]
  For $\theta \ge 0$ and $0 \leq q \leq 1$, the randomized
  $\theta$-threshold strategy with parameter $q$ is the
  randomized stopping rule that behaves as follows:
  it always selects the final candidate in the sequence
  if no earlier candidate was selected; otherwise, when
  observing candidate $i$ with value $v_i$, it never selects
  the candidate if $v_i < \theta$, it always selects the
  candidate if $v_i > \theta$, and when $v_i = \theta$ it
  selects the candidate with probability $q$.\footnote{If
  the same randomized tie-breaking rule was already invoked
  on a previous candidate $j$ with $v_j = \theta$ who was
  not selected, the probability of selecting candidate $i$
  in the present time step remains equal to $q$.}
\end{definition}

If $\tau_\lambda$ is a randomized $\theta$-threshold
strategy, the probability $\Pr(V(\tau_\lambda) \ge \theta)$
that the strategy selects an element whose value is
greater than or equal to $\theta$ always satisfies
the following bounds.
\begin{equation} \label{eq:rand-theta}
  \Pr(V^* > \theta) \le
  \Pr(V(\tau_\lambda) \ge \theta) \le
  \Pr(V^* \ge \theta) .
\end{equation}
Let $a$ and $b$, respectively,
denote the lower and upper bounds on the left and right
sides of~\eqref{eq:rand-theta}.
If $a=b$ --- i.e., if
the distribution of $V^*$ does not have a point-mass
at $\theta$ --- then all three of the quantities listed
in~\eqref{eq:rand-theta} must be equal, irrespective of
the value of the parameter $q$. Otherwise, as $q$ varies
from 0 to 1, the probability $\Pr(V(\tau_\lambda) \ge \theta)$
varies continuously and monotonically over the interval
$[a,b]$, starting at $a$ when
$q=0$ and ending at $b$ when
$q=1$. By the intermediate value theorem, for any
specified probability $\alpha$ in the interval $[a,b]$, we can
choose $q$ such that $\Pr(V(\tau_\lambda) \ge \theta) = \alpha$.

In the proofs of Theorem~\ref{thm-prophet} and
Proposition~\ref{random-order-prophet}, when we
write that $\theta$ is chosen so that
$\Pr(V^* > \theta) = \alpha$ for some specified
constant $\alpha$, what we really mean is that
\begin{equation} \label{eq:theta-inf}
  \theta = \inf \{ \theta' \mid \Pr(V^* > \theta') < \alpha \} .
\end{equation}
For this value of $\theta$, if we define
$a = \Pr(V^* > \theta)$ and $b = \Pr(V^* \ge \theta)$,
then $\alpha$ belongs to the interval $[a,b]$,
and consequently, as argued above,
there exists some $q \in [0,1]$ such
that the $\theta$-threshold strategy with
parameter $q$ has probability exactly $\alpha$
of selecting a candidate of value greater
than or equal to $\theta$. This particular
randomized $\theta$-threshold strategy is
the one analyzed in the proofs of
Theorem~\ref{thm-prophet} and
Proposition~\ref{random-order-prophet}.
When $a=b$ the choice of $q$ is immaterial
--- i.e., the equation $\Pr(V(\tau_\lambda) \ge \theta) = \alpha$
is satisfied regardless ---
so we adopt the convention that $q=0$
to accord with earlier sections of this paper,
in which $\theta$-threshold strategies
were assumed to select candidate $i$
only if the strict inequality $v_i > \theta$
is satisfied, or if $i$ is the last candidate.

Having thus defined $\theta$ and $\tau_\lambda$ in terms
of $\alpha$, the following properties are satisfied.
\begin{enumerate}
  \item For all $y < \theta$, $\Pr(V(\tau_\lambda) > y) \ge \alpha.$
  \item For all $y > \theta$, \[ \Pr(V(\tau_\lambda) > y) = \sum_{t=1}^n \Pr(v_t > y) \cdot
  \Pr(\tau_\lambda \mbox{ doesn't stop before } t) \ge (1-\alpha) \sum_{t=1}^n \Pr(v_t > y) . \]
\end{enumerate}
The first property holds simply because $\Pr(V(\tau_\lambda) > y) \ge \Pr(V(\tau_\lambda) \ge \theta) = \alpha$.
To justify the second property, first note that the event $V_t>y$ is
independent of the event that $\tau_\lambda$ doesn't stop before $t$,
since the former depends only
on the value $v_t$ whereas the latter depends only on $v_1,\ldots,v_{t-1}$.
Then, note that for any $t$,
\[
  \Pr(\tau_\lambda \mbox{ doesn't stop before } t) \ge
  \Pr(\tau_\lambda \mbox{ stops at } n) \ge \Pr(V(\tau_\lambda) < \theta) = 1-\alpha .
\]

The two properties listed above are the only properties of
$\theta,\, \tau_\lambda$ required by the proofs of
Theorem~\ref{thm-prophet} and Proposition~\ref{random-order-prophet},
so those proofs remain valid when we interpret the
$\theta$-threshold strategy as a randomized $\theta$-threshold
strategy with appropriately chosen parameter $q$.

\section{Missing Proofs From Section \ref{sec-var}} \label{app-order} We begin this appendix with a lemma showing that the parameter $\rho$ defined
in Proposition~\ref{random-order-prophet}, which bounds the ratio
$\expect[V^*]/\expect[V(\pi_\lambda)]$ when candidates are observed
in random order, is very close to $\ln(\lambda)$.

\begin{lemma} \label{lem:rho}
  For $\lambda>0$ the equation
  \begin{equation} \label{eq:rho2}
  \rho - \tfrac{1}{\lambda+1}(\rho-1) =
  \ln(\lambda+1) - \ln(1 - \rho^{-1})
  \end{equation}
  has a unique solution $\rho>1$.
  Considering this $\rho$ as a function
  of $\lambda$, it satisfies
  $|\rho - \ln(\lambda)| \to 0$
  as $\lambda \to \infty$.
\end{lemma}
\begin{proof}
  The left side of equation~\eqref{eq:rho2}
  is a strictly increasing continuous
  function of $\rho$ and
  the right side is a strictly decreasing
  continuous function of $\rho$ in the
  range $1 < \rho < \infty$.
  As $\rho \to 1$ the left side converges
  to 1 while the right side converges to
  $\infty$. As $\rho \to \infty$ the left
  side converges to $\infty$ while the
  right side converges to $\ln(\lambda+1)$.
  Since the difference between the two sides
  is a continuous, strictly monotonic function
  of $\rho$, there is a unique value of $\rho > 1$
  that equates the two sides.

  Rewriting equation~\eqref{eq:rho2} as
  \begin{equation} \label{eq:rho3}
    \rho  = \ln(\lambda+1) + \ln \left( \frac{1}{1 - \rho^{-1}} \right) + \frac{1}{\lambda+1} (\rho-1),
  \end{equation}
  we see that when $\rho>1$ all three of the quantities on the right side are positive,
  hence $\rho > \ln(\lambda+1) > \ln(\lambda)$. To bound the difference $\rho - \ln(\lambda)$
  we rewrite equation~\eqref{eq:rho3} as
  \begin{align} \nonumber
    \frac{\lambda}{\lambda+1} \cdot \rho & =
    \ln(\lambda+1) + \ln \left( \frac{1}{1 - \rho^{-1}} \right) - \frac{1}{\lambda+1} \\
    \nonumber
    \rho & = \frac{\lambda+1}{\lambda} \ln(\lambda+1) + \frac{\lambda+1}{\lambda} \ln \left( \frac{1}{1 - \rho^{-1}} \right) - \frac{1}{\lambda} \\
    \nonumber
         & = \ln(\lambda+1) + \frac{\ln(\lambda+1)}{\lambda} +
             \frac{\lambda+1}{\lambda} \ln \left( \frac{\rho}{\rho-1} \right) - \frac{1}{\lambda} \\
    \rho - \ln(\lambda) &= \ln \left( \frac{\lambda+1}{\lambda} \right) +
      \frac{\lambda+1}{\lambda} \ln \left( \frac{\rho}{\rho-1} \right) -
      \frac{1}{\lambda} .
      \label{eq:rho4}
    \end{align}
    The lower bound $\rho > \ln(\lambda)$ implies that
    as $\lambda$ tends to infinity, $\rho$ tends to infinity
    as well. All of the quantities on the right side
    of~\eqref{eq:rho4} --- namely, $\ln \left( \frac{\lambda+1}{\lambda} \right)$,
    $\frac{\lambda+1}{\lambda} \ln \left( \frac{\rho}{\rho-1} \right) $, and
    $\frac{1}{\lambda}$ --- converge to zero as $\lambda$ and $\rho$ both
    tend to infinity.
\end{proof}

For both of the following proofs recall that $V(\pi_\lambda(V_1,\ldots, V_n))$ is a random variables equal to the value of the candidate selected by the optimal $\lambda$-biased stopping rule for the sequence $(V_1,\ldots, V_n)$.

\begin{claim} \label{clm-2-point-tight}
There exists a family of instances with all candidates drawn from 2-point distributions, such that as $\lambda$ goes to infinity the ratio between $\expect[V^*]$ and $\expect[V(\pi_\lambda)]$ in the best ordering approaches $2$.
\end{claim}
\begin{proof}
We consider candidates with values drawn from the following distributions, for small $\eps>0$:
 	\begin{align*}
 V_1=
 \begin{cases}
 1,~~~~w.p~~ \eps \\
 \eps^2~~~~w.p~~ 1-\eps
 \end{cases}
 ,~~~
 V_2=
 \begin{cases}
 \eps+\eps^2(1-\eps),~~~~w.p~~ 1-\eps^2\\
 0,~~~~w.p~~ \eps^2
 \end{cases}
 \end{align*}

In this family of instances the expected value achieved by a prophet is:
\begin{align*}
\expect[V^*] &=\eps+(1-\eps)(1-\eps^2)\cdot ( \eps+\eps^2(1-\eps))+(1-\eps)\eps^4 \geq 2\eps -O(\eps^2)
\end{align*}
We show that in any ordering the expected payoff of a biased agent with $\lambda >1/\eps^{4}$ is $\expect[V(\pi_\lambda)] =\expect[V_1] = \eps+\eps^2(1-\eps) \leq \eps-\eps^2$. Thus, as $\eps$ goes to $0$ the ratio between $\expect[V^*]$ and $\expect[V(\pi_\lambda)]$ is approaching $2$. To see that $\expect[V(\pi_\lambda)] =\expect[V_1]$ observe that:
\begin{itemize}
	\item $\expect[V(\pi_\lambda(V_1,V_2))] = \expect[V_1]$ - We show that even when a low value is realized $V_1$ is selected:
	\begin{align*}
	\eps^2 \geq (1-\eps^2)(\eps+(1-\eps)\eps^2)-\lambda \cdot \eps^2\eps^2
	\end{align*}
	For $\lambda > 1/\eps^4$ the right hand-side is negative and hence the inequality holds.
	\item $\expect[V(\pi_\lambda(V_2,V_1))] = \expect[V_1]$  -  Observe that the high value of candidate $2$ equals $\expect[V_1]$. Thus, it suffices to show that the agents selects $V_2$ when a high value is realized. To see why this is the case observe that for any $\lambda>0$:
	\begin{align*}
	\expect[V_1] \geq \expect[V_1] -\lambda(1-\eps) (\eps+\eps^2(1-\eps) - \eps^2 )
	\end{align*}
\end{itemize}
\end{proof}

\end{document}